\documentclass{amsart}
\usepackage{algpseudocode}
\usepackage{algorithm}
\usepackage{amssymb}
\usepackage{stmaryrd}
\usepackage{color}
\usepackage{hyperref}
\usepackage[all]{xy}
\usepackage{tikz-cd}
\usepackage{rotating}

\newtheorem{theorem}{Theorem}[section]

\theoremstyle{definition}
\newtheorem{definition}[theorem]{Definition}

\newtheorem{remark}[theorem]{Remark}
\numberwithin{equation}{section}
\usepackage[all]{xy}

\newcommand{\LL}{\mathcal{L}}

\begin{document}

\title{A Digital signature scheme based on Module-LWE and Module-SIS}

\author{Huda Naeem Hleeb Al-Jabbari}
\address{Department of Mathematics, Tarbiat Modares University, 14115-134, Tehran, Iran}
\curraddr{} \email{h-aljbbari@modares.ac.ir}
\thanks{}

\author{Abbas Maarefparvar$^{*}$}
\address{Department of Mathematics and Computer Science, University of Lethbridge, Lethbridge, Canada}
\curraddr{}
\email{abbas.maarefparvar@uleth.ca}
\thanks{$^{*}$Corresponding author}
\date{}
\dedicatory{}
\commby{}



\begin{abstract}
In this paper,  we present an improved version of the digital signature scheme proposed by Sharafi and Daghigh \cite{sharafi2022ring} based on Module-LWE and Module-SIS problems. Our proposed signature scheme has a notably higher security level and smaller decoding failure probability, than the ones in the Sharaf-Daghigh scheme, at the expense of enlarging the module of the underlying basic ring.
\end{abstract}

\maketitle

 \section*{Notations}

We use bold lower-case letters to denote vectors, e.g., $\textbf{a}$, and use bold upper-case letters like $\textbf{A}$ to denote matrices. The concatenation of two vectors $\textbf{a}$ and $\textbf{b}$ is denoted by
$\textbf{a} || \textbf{b}$. %We denote by $<.,.>$ the usual dot product of two vectors
The uniform probability distribution over some finite set $S$ will be denoted by $U(S)$. If $s$ is sampled from a distribution $D$, we write $s \leftarrow D$. Logarithms are base $2$ if not stated otherwise. 
By default, all vectors will be column vectors, and for a vector $\textbf{v}$, we denote by $\textbf{v}^T$ its transpose. The boolean operator
$\llbracket \text{statement} \rrbracket$ evaluates to $1$ if statement is true, and to $0$ otherwise.
%, and sometimes identify a matrix with its ordered set
%of column vectors.
%We denote the horizontal concatenation of vectors and/or matrices using a vertical bar,
% e.g., $[\textbf{A} \, | \textbf{Ax}]$. We sometimes apply functions entry-wise to vectors, e.g., $\lfloor x \rceil$ rounds each entry of $\textbf{x}$ to its
 %nearest integer.

 \section{Introduction}

 %%%%%%%%%%%%%%%%%%%%%%%%%%%%%%%%%%%%%%%%%%%%%%%%%%
 %%%%%%%%%%%%%%%%%%%%%%%%%%%%%%%%%%%%%%%%%%%%%%%%%
 %%%%%%%%%%%%%%%%%%%%%%%%%%%%%%%%%%%%%%%%%%%%%%%%
 The advent of quantum computing poses significant challenges to the security of classical digital signature schemes. In anticipation of this paradigm shift, the cryptographic community has turned its focus towards post-quantum cryptography, seeking to develop algorithms that can withstand the computational prowess of quantum adversaries.

 In 2016, the National Institute of Standards and Technology (NIST) initiated a process to solicit, evaluate, and standardize one or more quantum-resistant public-key cryptographic algorithms \cite{alagic2022status}.  The submitted algorithms for the NIST PQC standardization are designed based on various hard computational problems, including lattices, codes, and hash functions, which are currently believed to resist quantum algorithm attacks.  Due to its rich number-theoretic structure, Lattice-Based Cryptography (LBC)  is one of the most promising alternatives among all the candidates. In particular,  most of the algorithms selected for NIST standardization, are lattice-based ones, namely Crystals-Kyber \cite{CRYSTALS-KYBER} (in the part ``Public-key Encryption and Key-establishment Algorithms'') and Crystals-Dilithium \cite{CRYSTALS-Dilithium} and FALCON \cite{Falcon-NIST} (in the part  ``Digital Signature Algorithms''), see \cite{alagic2022status}.

 Recently, inspired by the Lindner-Pikert cryptosystem \cite{lindner2011better}, Sharafi and Daghigh \cite{sharafi2022ring} designed a lattice-based digital signature whose security is based on the hardness assumption of the Ring Learning With Errors (Ring-LWE) and the Ring Short Integer Solution (Ring-SIS) problems, see Section \ref{section, lattice-based cryptography} for their definitions. In this paper, using the Module Learning With Errors problem (Module-LWE) and the Module Short Integer Solution problem (Module-SIS) we give some improvements to the Sharafi-Daghigh signature scheme. In particular, we show that using the Module-LWE/SIS would significantly increase the security of the algorithm. In addition, applying some implementation considerations we achieve much shorter public key and signature sizes than the ones in \cite{sharafi2022ring}.

 %In this paper, we use the number theoretic background of the lattice-based cryptography and provide a  digital signature scheme whose security will be proved in the quantum random oracle model.

%Lattices are geometric objects that can be pictorially described as the
%set of intersection points of an infinite, regular n-dimensional grid. Despite
%their apparent simplicity, lattices hide a rich combinatorial structure,
%%which has attracted the attention of great mathematicians over the
%last two centuries. Lattices have found numerous applications
%in mathematics and computer science, ranging from number
%theory and Diophantine approximation, to combinatorial optimization
%and cryptography.
%Lattice-based cryptography has received much attention from cryptographers and a comprehensive
%background on hard problems and security reductions exists.
%Lattice-based cryptography has emerged as a central area of research in the
%pursuit of designing quantum-safe primitives and advanced cryptographic constructions.

%Ajtai's framework has led  to use lattices in
%cryptography. More precisely, using Ajtai's results, one can 
%design cryptographic functions that are provably
%as hard as to break it is to solve a computationally hard lattice.

\section{lattice-based cryptography} \label{section, lattice-based cryptography}

%Throughout this article, we use the following notations:
The study of lattices, specifically from a computational point of view,
was marked by two major breakthroughs: the development of the LLL 
lattice reduction algorithm by Lenstra, Lenstra, and Lov\'asz in the early
80s \cite{lenstra1982factoring}, and Ajtai's discovery of a connection between the worst-case and
average-case hardness of certain lattice problems in the late 90's \cite{ajtai1996generating}. Recently, lattice-based cryptography has received much attention from cryptographers and a comprehensive
background on hard problems and security reductions exists. In this section, we introduce some basic notions in lattice-based cryptography and present some fundamental computational hard problems in this area. %The most part of this section is extracted from \cite[Section 2]{albrecht2017large}.

%\begin{definition} An $ n $-dimensional \textit{lattice} $ \LL $ is any subset of $ \mathbb{R}^n $ which satisfies in the following conditions.
	%\begin{itemize}
	%	\item $\LL$ is an additive subgroup. For every $x,y \in \LL$, we have $ x-y \in \LL.$ 
	%	\item every $x \in \LL$ has a neighborhood in $ \mathbb{R}^n $ in which x is the only lattice.
%	\end{itemize}
%\end{definition}
%%%%%%%%%%%%%%%%%%%%%%%%%%%
%\begin{example}
%	The integer set $ \mathbb{Z^n} $ is a lattice. For every lattice set $ \LL ,$ $ c \LL$ is a lattice, where $ c \in \mathbb{R}.$
%\end{example}

\begin{definition} \label{definition, Lattice}
For   $d \geq 1$ integer, a $d$-dimensional lattice $\mathcal{L}$ is a discrete subgroup of $\mathbb{R}^d$. A basis of the lattice $\mathcal{L}$ is a set $\textbf{B}=\{\textbf{b}_1,\dots,\textbf{b}_n\} \subseteq \mathbb{R}^d$ such that $\textbf{b}_i$'s are linearly independent vectors in $\mathbb{R}^d$ and all their integer combinations form $\mathcal{L}$:
\begin{equation} \label{equation, L(B)}
\mathcal{L}(\textbf{B})=\mathcal{L}(\textbf{b}_1,\dots,\textbf{b}_n)=\left\{\sum_{i=1}^{n} c_i \textbf{b}_i \, : \, c_i \in \mathbb{Z}, \, \forall i=1,\dots,n \right\}.
\end{equation}
\end{definition}

\begin{remark}
The integers $d$ and $n$ are called  \textit{dimension} and \textit{rank} of the lattice $\mathcal{L}$, respectively (Note that $n \leq d$). If $n=d$, the lattice $\mathcal{L}$ is called \textit{full rank}. Throughout the paper, we only
consider the full-rank ($n$-dimensional) lattices in $\mathbb{R}^n$.
\end{remark}

\begin{remark}
	A lattice basis $\textbf{B}$ is not unique. For a lattice $\mathcal{L}$   with basis $\textbf{B}$, and for every unimodular
	matrix $\textbf{U} \in \mathbb{Z}^{n \times n}$ (i.e., one having determinant $\pm 1$), $\textbf{B}.\textbf{U}$ is also a basis of $\mathcal{L}(\textbf{B})$.
	%% because $U. \mathbb{Z}^n=\mathbb{Z}^n$.
\end{remark}

\iffalse

\begin{definition}
The determinant $det(\mathcal{L})$ of  a lattice $\mathcal{L}$ is $|det(B)|$ for any basis $B$ of $\mathcal{L}$. 
\end{definition}

\begin{definition} \label{definition, dual}
The dual (sometimes called reciprocal) of a lattice $\mathcal{L} \subseteq \mathbb{R}^n$ is defined as 
\begin{equation*}
\mathcal{L}^{\perp}:=\{\textbf{w} : <\textbf{w},\mathcal{L}> \subseteq \mathbb{Z} \},
\end{equation*}
i.e., the set of points whose inner products with the vectors in $\mathcal{L}$ are all integers. It is straightforward to %verify
that $\mathcal{L}^{\perp}$ is a lattice. 
\end{definition}

\begin{definition}
	For a lattice $\LL$, the \textit{minimum distance} of $\LL$ is the length of a shortest lattice element:
	$$ \lambda_1({\LL}):= min_{v \in \LL \setminus \lbrace 0 \rbrace} ||v||,$$
	where $ ||.||$ is the Euclidean norm. In general, we define $ \lambda_i({\LL})=r,$ if 
	$ r $ is the smallest value such that $ \LL $ 
	has $ i $ independent vector of norm at most $ r. $
\end{definition}
Since a lattice $ \LL $ is an additive group, we have a quotient group $ \mathbb{R}^n/ \LL $ with cosets as follows.
$$ c+\LL=\lbrace c+v :v \in \LL \rbrace .$$
%%%%%%%%%%%%%%%%%%%%%%%%%%%%
\fi

\begin{definition}
	For a lattice $\LL$, the \textit{minimum distance} of $\LL$ is the length of a shortest lattice element:
	$$ \lambda_1({\LL}):= min_{\textbf{v} \in \LL \setminus \lbrace 0 \rbrace} ||\textbf{v}||,$$
	  where $ ||\textbf{v}||$ is the Euclidean norm of $\textbf{v}=(v_1,\dots,v_n)$, i.e., $||\textbf{v}||=\sqrt{\sum_{i=1}^{n} v_i^2}$.
\end{definition}

%Since the computational problems on lattices have the significant role in the cryptography, we now propose these problems. First, we recall the shortest vector problem which is known as SVP.

\subsection{Computational hard problems on lattices}
The shortest vector problem (SVP) and the closest vector problem (CVP) are two fundamental problems in lattices and their conjectured intractability is the foundation for a large number of cryptographic applications of lattices.
\begin{definition} \label{definition, SVP}
For a lattice $\mathcal{L}$ with basis $\textbf{B}$, the Shortest Vector Problem (SVP) asks to find a shortest nonzero lattice vector, i.e., a vector $\textbf{v} \in \mathcal{L}(\textbf{B})$ with $||\textbf{v}||=\lambda_1({\LL}(\textbf{B}))$. In the $\gamma$-approximate $\text{SVP}_{\gamma}$, for $\gamma \geq 1$, the goal is to find a shortest nonzero lattice vector $\textbf{v} \in \mathcal{L}(\textbf{B}) \backslash \{\textbf{0}\}$ of norm at most $||\textbf{v}|| \leq \gamma. \lambda_1(\mathcal{L}(\textbf{B}))$.
\end{definition}

\begin{definition} \label{definition, CVP}
For a lattice $\mathcal{L}$ with basis $\textbf{B}$ and a target vector $\textbf{t} \in \mathbb{R}^n$, the Closest Vector Problem (CVP) asks, to find a vector $\textbf{v} \in \mathcal{L}(\textbf{B})$ such that $|| \textbf{v}-\textbf{t} ||$ is minimized. In the $\gamma$-approximate $\text{CVP}_{\gamma}$, for $\gamma \geq 1$, the goal is to find a lattice vector $\textbf{v} \in \mathcal{L}$  such that $|| \textbf{v}-\textbf{t} || \leq \gamma . \text{dist}(\textbf{t},\mathcal{L}(\textbf{B}))$ where
\begin{equation*}
 \text{dist}(\textbf{t},\mathcal{L}(\textbf{B}))=\inf \{|| \textbf{w}-\textbf{t} || \, : \, \textbf{w} \in \mathcal{L}(\textbf{B}) \}.
\end{equation*}
\end{definition}

\begin{remark}
	One can show that SVP and CVP and their $\gamma$-approximate versions are NP-hard problems, see \cite{hanrot2011algorithms} for a survey on this subject.
\end{remark}

\subsection{Learning With Errors (LWE)} 

A fundamental problem in lattice-based cryptography is the Learning with
Errors problem (LWE).
The seminal work of Regev \cite{regev2005lattices} establishes reductions from standard
problems such as SVP in general lattices to LWE, suggesting that
LWE is indeed a difficult problem to solve. In particular, the ability to solve LWE
in dimension $n$ implies an efficient algorithm to find somewhat short vectors in
\textit{any} $n$-dimensional lattice. LWE is parameterized by positive integers $n$ and $q$ and an error distribution $\chi$ over $\mathbb{Z}$, which is usually taken to be a discrete Gaussian of width $\alpha \in (0,1)$.

\begin{definition}
	For a vector $\textbf{s} \in \mathbb{Z}_q^n$ called the secret, the LWE distribution $A_{\textbf{s},\chi}$ over $\mathbb{Z}_q^n \times \mathbb{Z}_q$ is sampled by choosing $\textbf{a} \in \mathbb{Z}_q^n$ uniformly at random, choosing $e \leftarrow \chi$, and outputting $(\textbf{a},b=<\textbf{s},\textbf{a}>+e \, \mathrm{mod} \, q)$, where $<\textbf{s},\textbf{a}>$ denotes the inner product of the vectors $\textbf{s}$ and $\textbf{a}$.
\end{definition}

\textit{Search} and \textit{decision} are two main types of the LWE problem. The first case is to find the secret given LWE samples, and the second case is to distinguish between LWE samples and uniformly random ones.
%%%%%%%%%%%%%%%%%%%%%%%%%%%%5
\begin{definition} \label{definition, Search-LWE}
	(\textbf{Search-$LWE_{n,q,\chi,m}$}).
	Suppose $m$ independent samples $(\textbf{a}_i,b_i)\in \mathbb{Z}_q^n \times \mathbb{Z}_q$ are drawn by using $A_{\textbf{s},\chi}$ for a uniformly random $\textbf{s} \in \mathbb{Z}_q^n$ (fixed for all samples), find $\textbf{s}$.
\end{definition}

\begin{definition} \label{definition, Decision-LWE}
	(\textbf{Decision -$LWE_{n,q,\chi,m}$}).
	Given $m$ independent samples $(\textbf{a}_i,b_i) \in \mathbb{Z}_q^n \times \mathbb{Z}_q$ where every sample is distributed according to either $A_{\textbf{s},\chi}$ for a uniformly random $\textbf{s} \in \mathbb{Z}_q^n$ (fixed for all samples), or the uniform distribution, distinguish which is the case (with non-negligible advantage).
\end{definition}
%Note that the both problems are easy to solve, unless the the error terms from $ \chi $ appears. Because by using Gaussian elimination we can efficiently recover s
%from LWE samples.

%Search-LWE can be seen as an average-case bounded-distance decoding (BDD) problem on a certain
%family of q-ary m-dimensional integer lattices: for LWE samples, the vector b is relatively close to
%exactly one vector in the LWE lattice
%$$ \LL(A):= \lbrace A^ts: s \in  \mathbb{Z}_q^n \rbrace+ q\mathbb{Z}^m $$
%and goal is to find that vector.

The concrete and asymptotic hardness of theLWE has recently been surveyed in \cite{albrecht2015concrete, herold2018asymptotic}. In particular, the hardness of the LWE can be reduced to the approximate SVP in the worst case. 

\subsection{Ring-LWE} \label{subsection, Ring-LWE}
%ing Learning With Errors problem (Ring-LWE) was introduced by Lyubashevsky  et al. in \cite{LPR10}, as a means of speeding up cryptographic constructions based on the Learning With Errors problem.  Similar to LWE, breaking certain instantiations of the Ring-LWE is provably at least as hard as quantumly solving \textit{any} instance of certain lattice problems, i.e., in the worst case. The underlying lattice problem Ring-LWE, is the approximate SVP on \textit{ideal lattices}, which are algebraically structured lattices corresponding to ideals in ring of integers of number fields.

The Ring Learning With Errors problem (Ring-LWE) is a variant of the Learning With Errors problem (LWE) that was introduced by Lyubashevsky et al. in \cite{LPR10} to achieve faster and more efficient cryptographic schemes. Like LWE, the security of Ring-LWE is based on the hardness of solving certain lattice problems in the worst case, even with quantum computers. The specific lattice problem that underlies Ring-LWE is the approximate Shortest Vector Problem (SVP) on ideal lattices, which are special types of lattices that have an algebraic structure related to the ring of integers of a number field.

Let $K$ be a number field, $\mathcal{O}_K$ its ring of integers, and $q \geq 2$ a rational integer. The search variant of Ring-LWE with parameters $K$ and $q$ consists in recovering a secret $s \in \mathcal{O}_K^{\vee}/q\mathcal{O}_K^{\vee}$ from arbitrarily many samples $(a_i,a_i.s+e_i)$, where 
\begin{equation} \label{equation, dual of Ok}
	\mathcal{O}_K^{\vee}=\{x \in K\, | \, \text{Tr}_{K/\mathbb{Q}}(xy) \in \mathbb{Z}, \, \forall y \in \mathcal{O}_K\}
\end{equation}
 denotes the dual of $\mathcal{O}_K$ (by $\text{Tr}_{K/\mathbb{Q}}(a)$ we mean the \textit{trace} of $a \in K$ over $\mathbb{Q}$), each $a_i$ is uniformly sampled in $\mathcal{O}_K/q \mathcal{O}_K$ and each $e_i$ is a small random element of $K_{\mathbb{R}}:=K \otimes_{\mathbb{Q}} \mathbb{R}$. The decision variant of Ring-LWE consists in distinguishing arbitrarily many such pairs for a common secret $s$ chosen uniformly random in $\mathcal{O}_K^{\vee}/q \mathcal{O}_K^{\vee}$, from uniform samples in $\mathcal{O}_K^{\vee}/q\mathcal{O}_K^{\vee} \times K_{\mathbb{R}}/q \mathcal{O}_K^{\vee}$ \cite{rosca2018ring}.

\begin{remark}
	 The ring $R=\mathcal{O}_K$ is typically taken to be a power-of-two \textit{cyclotomic ring}, i.e., $R=\mathbb{Z}[\zeta_n]$ where $\zeta_n$ is a primitive $n^{\text{th}}$-root of unity for  $n=2^k$ a power of $2$. This has led to practical implementation advantages and a simpler interpretation of formally defined Ring-LWE, see \cite{LPR10,LPR13, lyubashevsky2013toolkit,peikert2016not}
 for more details on this subject.
\end{remark}

\subsection{Module-LWE}
The Module Learning With Errors problem (Module-LWE) is a generalization of the Ring Learning With Errors problem (Ring-LWE) that allows for more flexibility and efficiency in lattice-based cryptography \cite{brakerski2014leveled, LaSt15}.  Module-LWE is based on the hardness of finding small errors in linear equations over modules, which are collections of ring elements with a common structure. Module-LWE can be seen as a way of interpolating between LWE and Ring-LWE, where the module rank (the number of ring elements in each module) determines the level of security and performance. In particular, Ring-LWE is a special case of Module-LWE with module rank $1$ \cite{boudgoust2023hardness}.
%The Module Learning With Errors problem (Module-LWE) \cite{brakerski2014leveled, LaSt15} is a core computational assumption of lattice-based cryptography which offers an interesting trade-off between guaranteed security and concrete efficiency \cite{boudgoust2023hardness}. Indeed Module-LWE  was proposed to address shortcomings in both LWE and Ring-LWE by interpolating between them. Roughly speaking, the Module-LWE problem can be thought as taking the Ring-LWE problem and replacing the single ring elements  with module ones over the same ring. Based on this
%viewpoint, Ring-LWE can be thought as Module-LWE with module rank $1$.  

%%Following \cite{albrecht2017large} we restate the formal definition of the Module-LWE. ; Let $K$ be a number field, $\mathcal{O}_K$ its ring of integers, $q \geq 2$ be a rational integer, 

In order to give the formal definition of the Module-LWE, we have to fix some notations. Similar to the Ring-LWE, let $K$ be a number field of degree $n$, $R$ be the ring of integers of $K$, $R^{\vee}$ be the dual of $R$, as defined in \eqref{equation, dual of Ok}, and $q \geq 2$ be integer. Also let $K_{\mathbb{R}}:=K \otimes_{\mathbb{Q}} \mathbb{R}$,  $\mathbb{T}_{R}:=K_{\mathbb{R}}/R^{\vee}$, $R_q:=R/(qR)$ and $R_q^{\vee}:=R^{\vee}/(qR^{\vee})$.

\begin{definition} \label{definition, Module-LWE distribution}
	(\textbf{Module-LWE distribution, statement from \cite{albrecht2017large}}) Let $M:=R^d$. For $\textbf{s} \in (R_q^{\vee})^d$ and an error distribution $\psi$ over $K_{\mathbb{R}}$, we sample  the module learning with error distribution $A_{d,q,\textbf{s},\psi}^{M}$ over $(R_q)^d \times \mathbb{T}_{R^{\vee}}$ by outputting $(\textbf{a}, \frac{1}{q} <\textbf{a},\textbf{s}>+e \, \, \mathrm{mod} \,  R^{\vee})$ where $ a \leftarrow U((R_q)^d)$ and $e \leftarrow \psi$. 
	
\end{definition}

\begin{definition} \label{definition, Module-LWE}
	(\textbf{Decision/Search Module-LWE problem, statement from \cite{albrecht2017large}}) Let $\Psi$ be a family of distributions over $K_{\mathbb{R}}$ and $D$ be distribution over $R_q^{\vee}$. For	
	$M=R^d$,  the decision module learning with errors problem $\text{Module-LWE}_{m,q,\Psi}^{(M)}$ entails distinguishing $m$ samples of $U((R_q)^d) \times \mathbb{T}_{R^{\vee}}$ from $A_{q,\textbf{s},\psi}^{(M)}$ where $\textbf{s} \leftarrow D^d$ and $\psi$ is an arbitrary distribution in $\Psi$.
\end{definition}

\begin{remark}
Let $[K:\mathbb{Q}]=r_1+2r_2$, where $r_1$ and $r_2$ denote the number of real and complex (non-real) embedding of $K$, respectively.	One can show that  
	\begin{equation} \label{equation, KR is isomorphic to H}
K_{\mathbb{R}} \simeq	H:=\{\textbf{x} \in \mathbb{R}^{r_1} \times \mathbb{C}^{2r_2} \, : \, x_i=\overline{x_{i+r_2}} \, , \, i=r_1+1, \dots, r_1+r_2\},
	\end{equation}
where $\overline{x_{i+r_2}}$ denotes the complex conjugate of $x_{i+r_2}$. Hence the distribution over $K_{\mathbb{R}}$ are sampled by choosing an element of the space $H$ according to the distribution and mapping back to $K_{\mathbb{R}}$ via the isomorphism \eqref{equation, KR is isomorphic to H}, see \cite[Section 2.3]{albrecht2017large}. 
\end{remark}

Module-LWE comes with hardness guarantees given by lattice problems based on a certain class of lattices, called  \textit{module lattices}. % In this case, the lattices are generated by
%modules as opposed to ideals in the Ring-LWE case and in contrast to Ring-LWE, it has
%been shown that Module-LWE is equivalent to \textit{natural} hard problems over these lattices.
  In this case, solving the approximate SVP on module lattices
for polynomial approximation factors would permit solving Module-LWE (and thus Ring-LWE) efficiently \cite{albrecht2017large}.

\begin{remark}
    Module-LWE is a promising alternative to Ring-LWE that can resist potential attacks that exploit the algebraic structure of the rings \cite{CDW17}. Therefore, Module-LWE may provide a higher level of security than Ring-LWE, while still being more efficient than plain LWE. Moreover, Module-LWE allows for more fine-grained control over the security level by adjusting the module rank, which is not possible with efficient Ring-LWE schemes. Furthermore, Module-LWE implementations can be easily adapted to different security levels by reusing the same code and parameters \cite{albrecht2017large}.
\end{remark}

%Module-LWE has been suggested as an interesting option to hedge
%against potential attacks exploiting the algebraic structure of \cite{CDW17}.
%Thus, Module-LWE might be able to offer a better level of security than Ring-LWE, while
%still offering performance advantages over plain LWE. In addition, Module-LWE
%promises to adjust the security level with much greater granularity than efficient
%Ring-LWE instantiations and implementations for one security level can easily be
%extended to other security levels \cite{albrecht2017large}.

\subsection{SIS, Ring-SIS, and Module-SIS problems}
The Short Integer Solution (SIS) problem, introduced by Ajtai \cite{ajtai1996generating}, serves as one of the foundations of numerous lattice-based cryptographic protocols. The SIS problem is to find a short nonzero solution $\textbf{z} \in \mathbb{Z}^m$, with $0 < ||\textbf{z} || \leq \beta$, to the homogeneous linear system $\textbf{A}\textbf{z} \equiv \textbf{0} \, \mathrm{mod}\, q$ for uniformly random $\textbf{A} \in \mathbb{Z}^{m \times n}$, where $m,n,q$ denote positive integers and $\beta \in \mathbb{R}$. Inspired by the efficient NTRU encryption scheme \cite{NTRU98}, Micciancio \cite{micciancio2002generalized, micciancio2007generalized} initiated an approach that consists of changing the SIS problem to variants involving structured matrices. This approach was later replaced by a more powerful variant referred to as the Ring Short Integer Solution (Ring-SIS) problem \cite{lyubashevsky2006generalized, peikert2006efficient}. There are several reductions from some hard lattice problems to SIS and Ring-SIS, see \cite[Section 3]{LaSt15}. %Viewing the module structure as a generalized structure of a ring, one can extend the Ring-SIS problem to the module lattices, which is termed the Module Short Integer Solution (Module-SIS) problem \cite{LaSt15}.
The Module Short Integer Solution (Module-SIS) problem is a generalization of the Ring Short Integer Solution (Ring-SIS) problem that involves finding short vectors in module lattices, which are lattices with a module structure. %A module is a collection of ring elements that share a common algebraic property, such as being multiples of a fixed element. 
By viewing the rings as modules of rank $1$, we can extend the definition of Ring-SIS to Module-SIS and study its hardness and applications.

\begin{definition} \label{definition, Module-SIS}
	Let $K$ be a number field with ring of integers $R$, and $R_q:=R/qR$ where $q \geq 2$ is an integer. For positive numbers $m,d \in \mathbb{Z}$ and $\beta \in \mathbb{R}$, the problem $\text{Module-SIS}_{q,n,m,\beta}$  is to find $z_1,\dots,z_m \in R$ such that  $ 0<  ||(z_1,\dots,z_m) || \leq \beta$ and $\sum_{i=1}^{m} \textbf{a}_i.z_i \equiv \textbf{0} \, \mathrm{mod}\, q$, where  $\textbf{a}_1,\dots, \textbf{a}_m \in R_q^d$ are chosen independently from the uniform distribution.
\end{definition}

\begin{remark}
	Note that $N=nd$ denotes the dimension of the corresponding module lattice and gives the complexity statements for $N$ growing to infinity, where $n$ denotes the extension degree of $K$ over $\mathbb{Q}$ \cite[$\S$3]{LaSt15}. For the worst-case to average-case reduction from computational hard lattice problems to Module-SIS problems, see \cite[Theorem 4]{LaSt15}.
\end{remark}

%\section{Related Works}

%In this section .... .

\section{Our Contribution}
The worst-case to average-case reductions for the module problems
are (qualitatively) sharp, in the sense that there exist converse reductions. This property is
not known to hold in the context of Ring-SIS/Ring-LWE \cite{LaSt15}. In addition, using these module problems, one could design cryptographic algorithms whose key sizes are notably smaller than similar ring-based schemes, see e.g.  \cite{CRYSTALS-Dilithium, CRYSTALS-KYBER}.
 %due to the aforementioned advantages of Module-LWE and Module-SIS (in comparison to their ring settings), 
 In this section,  we present an improved version of the digital signature proposed in \cite{sharafi2022ring} based on Module-LWE and Module-SIS problems. Similar to the Sharafi-Daghigh scheme \cite{sharafi2022ring}, the structure of our digital signature algorithm is based on the ``hash-and-sign'' approach and ``Fiat-Shamir paradigm'' \cite{Lyu09, Lyu12}. Our contribution to the proposed signature scheme can be summarized as follows:
 \begin{itemize}
 	\item 
Using Module-LWE and Module-SIS problems, instead of Ring-LWE and Ring-SIS respectively,  would increase the security levels at the cost of slightly increasing the key sizes, see Section \ref{section, parameters}.
 	
 	\item 
 Inspired by some celebrated lattice-based algorithms \cite{CRYSTALS-Dilithium, CRYSTALS-KYBER}, we generate the public key $\textbf{A}$ from a $256$-bits \text{seed} $\zeta$ as input of a hash function $G$. Based on this standard technique, instead of $n \log q$ bits as in  \cite{sharafi2022ring}, one needs only $256$ bits to store the public key $\mathbf{A}$.

 	\item 
The encode and decode functions (for serializing a polynomial in byte arrays and its inverse) used in \cite{sharafi2022ring} are the straightforward methods introduced in Lindner-Peikert scheme \cite{lindner2011better}.  In this paper,  as a modification,  we use the \textit{NHSEncode} and \textit{NHSDecode} functions given in  NewHope-Simple \cite{alkim2016newhope}.  This method will decrease  the decoding failure rate significantly, see Sections \ref{section, Decoding failure probability} and \ref{section, parameters}.
 	
 	\item 
 	Following Crystals-Kyber \cite{CRYSTALS-KYBER} and NewHope \cite{NewHope,NewHope-USENIX,alkim2016newhope}, we use the centered binomial distribution, instead of the Gaussian one as used in \cite{sharafi2022ring}, to sample noise and secret vectors in the proposed signature scheme. This gives a  more efficient implementation than \cite{sharafi2022ring} and increases the security against side-channel attacks.
 \end{itemize}
  
% We show that this new setting would increase the security level and decrease the key sizes significantly. 

\section{Proposed Digital Signature}

In this section, we present a Module-LWE based version of Sharafi-Daghigh digital signature \cite{sharafi2022ring}. It must be noted that we only provide a module version of the signature scheme and its efficient implementation would be an interesting challenge for future works. So we will describe each step by pseudocodes, however, one can apply the corresponding auxiliary functions of any Module-LWE based algorithm, e.g. Crystals schemes \cite{CRYSTALS-Dilithium,CRYSTALS-KYBER}, to gain a test implementation.

\subsection{Encoding and Decoding functions} \label{section, encoding-decoding} In \cite{sharafi2022ring} in order to encode a bit array $\nu=(\nu_0,\dots,\nu_{n-1}) \in \{0,1\}^n$ into a polynomial  $\textbf{v} \in R_q$ and its inverse, the authors use the following  methods as proposed by Lindner and Peikert \cite{lindner2011better}:
\begin{align*}
&\textbf{v}=\text{Encode}(\nu_0,\dots,\nu_{n-1})=\sum_{i=0}^{n-1} \nu_i.\lfloor \frac{q}{2}\rfloor . X^i, \\
&\mu=(\mu_0,\dots,\mu_{n-1})=\text{Decode}\left(\sum_{i=0}^{n-1} v_i . X^i\right), \quad \text{where} \quad \mu_i=\left\{
\begin{array}{rl}
	1 & \text{if} \,  v_i \in \left[- \lfloor \frac{q}{4}\rfloor,\lfloor \frac{q}{4}\rfloor \right) \\
	0 & \text{otherwise}. \\
\end{array} \right.
\end{align*}

As mentioned before, one of our contributions is to use different encoding-decoding functions leading to a much lower error rate. We use the \textit{NHSEncode} and \textit{NHSDecode} functions introduced in \textit{NewHope-Simple} \cite{alkim2016newhope}. Moreover, sampling secret and error polynomials from a \textit{binomial distribution} enables us to apply the same analysis as \textit{NewHope} \cite{NewHope-USENIX,alkim2016newhope} to obtain an error rate bounded by $2^{-60}$. In \textit{NHSEncode}, each bit of $\nu \in \{0,1\}^{256}$ is encoded into four coefficients. The decoding function \textit{NHSDecode} maps from four coefficients back
to the original key bit; %First map the four coefficients
%from the range $\{0,\dots,q-1\}$ into the range $\{-\lfloor \frac{q}{2} \rfloor,\dots, \lfloor \frac{q}{2} \rfloor \}$, accumulate their absolute values,
%and then set the key bit to $1$ if this sum is larger than $q$ and to $0$ otherwise. 
Take four coefficients (each in the range  $\{0,\dots,q-1\}$), subtract
$\lfloor \frac{q}{2} \rfloor$ from each of them, accumulate their absolute values and set the key bit to $0$ if the sum is larger
than $q$ or to $1$ otherwise. See Algorithms $1$ and $2$ in \textit{NewHope-Simple} \cite{alkim2016newhope} for more details.

\subsection{Key Generation} In order to decrease size of the public key, in comparison to \cite{sharafi2022ring}, one can use a \textit{seed} (as input parameter of informal function $\text{GenA}$) to generate uniformly public matrix $\hat{\textbf{A}} \in R_q^{k \times k}$ in NTT domain. To accomplish this, for instance, the key generation algorithms of Crystals schemes \cite{CRYSTALS-Dilithium,CRYSTALS-KYBER} could be applied. Likewise, the secret and error vectors $\textbf{s}, \textbf{e} \in R_q^k$ could be sampled from a seed. Moreover, in order to ease implementation, $\textbf{s}$ and $\textbf{e}$ would be drawn from the binomial distribution with standard derivation $\eta$ instead of the Gaussian one. Following Crystals-Dilithum \cite{CRYSTALS-Dilithium}, we use the notation $S_{\eta}$ to denote the subset of $R_q$ consisting of all polynomials whose coefficients have size at most $\eta$. Also, we denote the corresponding informal function by $\text{GenSE}$. The pseudocode for the key generation is described in Algorithm \ref{alg:key gen}.

\subsection{Signing} In order to generate a signature $\sigma$ assigned to a message $M$, the signer first uses a collision-resistant hash function, denoted by $\text{CRH}$, and computes the hash of $M$, say $\mu$. Then using a random $256$-bits coin $r$, as input value of the function $\text{GenSE}$, she generates the vectors $\textbf{e}_1,\textbf{e}_2 \in R_q^k$ and $e_3,e_4 \in R_q$ from a binomial distribution of parameter $\eta$. The signature $\sigma$ would contain the following four components 
\begin{align*}
\textbf{z}_1&=\textbf{A} \textbf{e}_1+\textbf{e}_2 \in R_q^k, \\
z_2&=\textbf{P}\textbf{e}_2+e_4 \in R_q, \\
z_3&= \textbf{A} \textbf{P} \textbf{e}_1+e_3+\text{NHSEncode}(\mu) \in R_q, \\
h&=\text{CRH}(\mu || \text{CRH}(\text{NHSDecode}(\textbf{A} \textbf{s} \textbf{e}_2))).
\end{align*}
%where $\text{NHSEncode}$ and $\text{NHSDecode}$ denote the encode and decode functions defined in \textit{NewHope-Simple} \cite{alkim2016newhope}.
 (Recall that the functions $\text{NHSEncode}$ and $\text{NHSDecode}$ are introduced in Section \ref{section, encoding-decoding} and $\textbf{s}$ denotes the secret key). The pseudocode for the signing step is described in Algorithm \ref{alg:signing}.

\subsection{Verification} As described in Algorithm \ref{alg:verify}, using the public key $\textbf{P}$, the verifier accepts $\sigma=(\textbf{z}_1,z_2,z_3,h)$ as a valid signature of the message $M$ if and only if both of the following conditions hold:
\begin{itemize}
	\item[(1)]
$\text{NHSDecode}(z_2+z_3-\textbf{P}\textbf{z}_1)=\text{CRH}(M)$,
	\item[(2)]
	$h=\text{CRH}(\text{CRH}(M) || \text{CRH}(\text{NHSDecode}(z_2))$.
\end{itemize}

	\begin{algorithm}
	\caption{key generation}\label{alg:key gen}
	\begin{algorithmic}[1]
		\State  \textbf{Output}: Public key $pk \in \{0,1\}^{256} \times R_q^k$
		\State  \textbf{Output}: Secret key $sk \in R_q^k$
		%\Require $n \geq 0$
		%\Ensure $y = x^n$
		\State $\zeta \leftarrow \{0,1\}^{256}$
		\State $(\rho , \xi) \in \{0,1\}^{256 \times 2}  := H(\zeta)$ \Comment{$H$ is instantiated as SHAKE-256 throughout}
		\State $\textbf{A} \in R_q^{k \times k} :=\text{GenA}(\rho)$  \Comment{$\textbf{A}$ is stored in NTT representation as $\hat{\textbf{A}}$}
		\State $(\textbf{s},\textbf{e}) \in S_{\eta}^{k} \times S_{\eta}^{k}:=\text{GenSE}(\xi)$ \Comment{Sample $\textbf{s} \in R_q^k$  and $\textbf{e} \in R_q^k$ from a binomial distribution of parameter $\eta$}
		
		\State  $\textbf{P}:=\text{NTT}^{-1}(\hat{\textbf{A}}^T \circ\text{NTT}(\textbf{s}))+\textbf{e}$ \Comment{$\textbf{P}:=\textbf{A}^T\textbf{s}+\textbf{e} \in R_q^k$}
	%	\While{$N \neq 0$}
	%	\If{$N$ is even}
	%	\State $X \gets X \times X$
	%	\State $N \gets \frac{N}{2}$  \Comment{This is a comment}
	%	\ElsIf{$N$ is odd}
	%	\State $(s,e) \in S_{\eta}^{\ell} \times S_{\eta}^{k}$
	%	\State $N \gets N - 1$
	%	\EndIf
	%	\EndWhile
	\State  \textbf{Return} \, $(pk=(\rho,\textbf{P}),sk=\textbf{s}$
	\end{algorithmic}
\end{algorithm}

	\begin{algorithm}
	\caption{Signing}\label{alg:signing}
	\begin{algorithmic}[1]
		\State \textbf{Input:}  Message $M$
			\State \textbf{Input:}  Public key $pk=(\rho,\textbf{P})$
		\State \textbf{Input:}  Secret key $\textbf{s}$
		\State \textbf{Input:}  Random coin $r \in \{0,1\}^{256}$
		\State \textbf{Output:} Signature $\sigma \in R_q^k \times R_q \times R_q \times \{0,1\}^{256}$	
		%\Require $n \geq 0$
		%\Ensure $y = x^n$
		\State $\mu:=\text{CRH}(M) \in \{0,1\}^{256}$
		\State $\textbf{A} =\text{GenA}(\rho)$ 
		\State $(\textbf{e}_1,\textbf{e}_2,e_3,e_4) \in S_{\eta}^k \times S_{\eta}^k \times S_{\eta}^1 \times S_{\eta}^1:=\text{GenSE}(r)$ \Comment{Sample $\textbf{e}_1, \textbf{e}_2 \in R_q^k$ and $e_3,e_4 \in R_q$ from a binomial distribution  of parameter $\eta$}
	
		\State $\textbf{z}_1=\text{NTT}^{-1}(\hat{\textbf{A}}^T \circ \text{NTT}(\textbf{e}_1)) + \textbf{e}_2$  \Comment{$\textbf{z}_1:=\textbf{A} \textbf{e}_1+\textbf{e}_2 \in R_q^k$}
		
		\State $z_2=\text{NTT}^{-1}(\hat{\textbf{P}}^T \circ \text{NTT}(\textbf{e}_2)) + e_4$
		\Comment{$z_2:=\textbf{P}^T
			 \textbf{e}_2+e_4 \in R_q$}
	\State $z_3=\text{NTT}^{-1}\left((\hat{\textbf{A}} \circ \hat{\textbf{P}} )^T \circ \hat{\textbf{e}_1}\right) + e_3+\text{NHSEncode}(\mu)$ 	\Comment{$z_3:= \textbf{A} \textbf{P} \textbf{e}_1+e_3+\text{NHSEncode}(\mu) \in R_q$}
	\State $h:=\text{CRH}\left(\mu || \text{CRH}\left(\text{NHSDecode}\left(\text{NTT}^{-1}((\hat{\textbf{A}}\circ \text{NTT}(\textbf{s}))^T \circ \hat{\textbf{e}_2}) \right) \right)\right)$ \Comment{$h:=\text{CRH}(\mu || \text{CRH}(\text{NHSDecode}(\textbf{A} \textbf{s} \textbf{e}_2)))$}
		\State  \textbf{Return} \, $\sigma=(\textbf{z}_1,z_2,z_3,h)$
		%\While{$s$}
		%\If{$N$ is even}
	%	\State $X \gets X \times X$
	%	\State $N \gets \frac{N}{2}$ 
	%	\ElsIf{$N$ is odd}
	%	\State $y \gets y \times X$
	%	\State $N \gets N - 1$
	%	\EndIf
	%	\EndWhile
	\end{algorithmic}
\end{algorithm}

	\begin{algorithm}
	\caption{Verifying}\label{alg:verify}
	\begin{algorithmic}[1]
			\State \textbf{Input:}  Message $M$
		\State \textbf{Input:}  Public key $pk=(\rho,\textbf{P})$
		\State \textbf{Input:} Signature $\sigma=(\textbf{z}_1,z_2,z_3,h)$		
			\State $\mu=\text{CRH}(M)$
		%\State $\textbf{A}=\text{GenA}(\rho)$ 
		%\State $\hat{\textbf{p}}:=\text{NTT}(\textbf{P})$
		%\State $\hat{\textbf{z}_1}:=\text{NTT}(\textbf{z}_1)$
		\State $w=z_2+z_3-\text{NTT}^{-1}\left(\text{NTT}(\textbf{P})^T \circ \text{NTT}(\textbf{z}_1) \right)$ \Comment{$w=z_2+z_3-\textbf{P} \textbf{z}_1 \in R_q$}
		\Return $\llbracket \text{NHSDecode}(w)=\mu \rrbracket$ \quad and $\llbracket h=\text{CRH}(\mu || \text{CRH}(\text{NHSDecode}(z_2)) \rrbracket$
	\end{algorithmic}
\end{algorithm}

\begin{remark}  \label{remark, formula for size}
	Note that the publick key $(\textbf{A},\textbf{P})$ has size $256+kn\log q$ bits or equivalently its size would be $32+\frac{kn\log q}{8}$ bytes. Likewise, the secret key $\textbf{s} \in R_q^k$ has size $\frac{kn\log q}{8}$ bytes, and the size of the signature $\sigma$ is $\frac{(k+2)n\log q}{8}+32$ bytes. The corresponding sizes for each parameter set are given in Table \ref{Table, security estimating}.
\end{remark}

\subsection{Decoding failure probability} \label{section, Decoding failure probability} Similar to computations as in \cite[Section 3.1]{sharafi2022ring}, we have
\begin{align*}
	z_2+z_3-\textbf{P}\textbf{z}_1&=\textbf{P}\textbf{e}_2+e_4+\textbf{A}\textbf{P}\textbf{e}_1+e_3+\text{NHSEncode}(\mu)-\textbf{P}\left(\textbf{A}\textbf{e}_1+\textbf{e}_2\right) \\
	&=\text{NHSEncode}(\mu)+\left(e_3+e_4\right), \\
	z_2&=\left(\textbf{A}\textbf{s}+\textbf{e}\right)\textbf{e}_2+e_4=\textbf{A}\textbf{s}\textbf{e}_2+\left(\textbf{e}\textbf{e}_2+e_4\right).
\end{align*}

Hence applying $\text{NHSDecode}$ on the above terms, the verifier obtains respectively $\mu$ and $\textbf{A}\textbf{s}\textbf{e}_2$ as long as the error terms $e_3+e_4$ and $\textbf{e}\textbf{e}_2+e_4$ have been sufficiently small sampled. But following NewHope \cite{NewHope-USENIX,alkim2016newhope} all the polynomials $\textbf{e},\textbf{e}_1,\textbf{e}_2,e_3,e_4$ are drawn from the centered binomial distribution $\psi_{\eta}$ of parameter $\eta=16$. Since the decoding function $\text{NHSDecode}$ maps from four coefficients to one bit, the rounding error for every $4$-dim chunk has size at most $q/4+4$. As mentioned in NewHope-Simple \cite{alkim2016newhope}, one can use the same analysis as in NewHope-USENIX \cite[Appendix D]{NewHope-USENIX} to conclude that the total failure rate in encoding-decoding process is at most $2^{-60}$. In Section \ref{section, parameters}, we will choose the parameter sets for the proposed scheme to achieve this failure probability.
(Note that the upper bound of the failure rate in Sharafi-Daghigh scheme \cite{sharafi2022ring} is $2^{-40}$).

\section{Unforgeability in Random Oracle Model} \label{section, Unforgeability}

The proposed digital signature scheme is a module-LWE-based version of the Sharafi-Daghigh algorithm \cite{sharafi2022ring}. 
Hence replacing ring-based hard problems with corresponding module-based ones, one can transmit the security proof in \cite[Section 4]{sharafi2022ring} to the module-based setting. However, here following Crystals-dilithium \cite[Section 6]{CRYSTALS-Dilithium}, we give a more detailed proof, including a simulation of signature queries, for the SUF-CMA security (Strongly Unforgeability under Chosen Message Attack\footnote{In this security model, the adversary gets the public key and has access to a signing oracle to sign messages of his choices. The adversary's goal is to produce a valid new signature for a previously signed message \cite{boneh2006strongly}.})  of the proposed signature scheme in the \textit{Classical Random Oracle Model (ROM)}.

%	This means that an adversary who is given a signature for a few messages of his choice should not be able to produce a \textit{valid} signature for a new message  \cite{boneh2006strongly}.}) of the proposed signature scheme in the \textit{Classical Random Oracle Model (ROM)}.

\begin{theorem} \label{theorem, security}
	Assume that $H:\{0,1\}^* \rightarrow R_q$ is a cryptographic hash function modeled as a random oracle. If there exists an adversary $\mathcal{A}$ (who has classical access to $H$) that can break the UF-CMA security of the proposed signature, then there exist also adversaries $\mathcal{B}$ and $\mathcal{C}$ such that
	\begin{equation*}
		\text{Adv}^{\text{UF-CMA}}(\mathcal{A}) \leq \text{Adv}_{k,\eta}^{\text{MLWE}}(\mathcal{B}) + \text{Adv}_{k,\eta}^{\text{H-MSIS}}(\mathcal{C}),
	\end{equation*}
where
\begin{align*}
\text{Adv}_{k,\eta}^{\text{MLWE}}&:=|\text{Pr}[b=1| \textbf{A} \leftarrow R_q^{k \times k}; \textbf{P} \leftarrow R_q^k; b \leftarrow \mathcal{A}(\textbf{A},\textbf{P})] \\ 
&-\text{Pr}[b=1 | \textbf{A} \leftarrow R_q^k; \textbf{s} \leftarrow S_{\eta}^1; \textbf{e}\leftarrow S_{\eta}^1; b \leftarrow \mathcal{A}(\textbf{A},\textbf{A}\textbf{s}+\textbf{e})]|,
\end{align*}
and
\begin{align*}
&\text{Adv}_{k,\eta}^{\text{H-MSIS}}=\\
&\text{Pr}\left[0<||z||_{\infty} \leq \eta \,\wedge H(M||\textbf{A}z)=h^{\prime} | \textbf{A} \leftarrow R_q^{k \times k}; (z,h^{\prime},M) \leftarrow \mathcal{A}^{|H(.)>}(\textbf{A}) \right].
\end{align*}
\end{theorem}

\begin{proof}
The security of our proposed signature  algorithm is based on the decision Module-LWE  and the \textit{combined} hardness of Module-SIS with the hash function $H$. By the decision Module-LWE assumption, which  denoted as $\text{Adv}_{k,\eta}^{\text{MLWE}}$, the public key $(\textbf{A},\textbf{P})$ (a module-LWE sample) is indistinguishable from the pair $(\textbf{A}^{\prime},\textbf{P}^{\prime})$ chosen uniformly at random. This guarantees the security against the key recovery attack. To prove the security in the SUF-CMA model, let the adversary gets the public key $(\textbf{A},\textbf{P})$ and successfully produces a signature $\sigma^{\prime}=(\textbf{z}_1^{\prime},z_2^{\prime},z_3^{\prime},h^{\prime})$ for the message $M$ which is \textit{valid} according to  the Verification Algorithm \ref{alg:verify}.
	In particular, the adversary would be able to find  $z_2^{\prime} \in S_{\eta}^1$ such that
	\begin{equation*}
		H(\text{NHSDecode}(z_2^{\prime}))=H(\text{NHSDecode}(\textbf{A}\textbf{s}\textbf{e}_2)).
	\end{equation*}  
	%%\begin{itemize}
	%	\item[(1)] $\text{NHSDecode}(z_2^{\prime}+z_3^{\prime}-\textbf{P}\textbf{z}_1^{\prime})=H(M)$;
	
	%	\vspace*{0.2cm}
	%	\item[(2)]
	%	$h^{\prime}=H(H(M) || H(\text{NHSDecoode}(z_2^{\prime})))$.
	%\end{itemize}
	Therefore to create a forged signature of the message M, the adversary will encounter two challenges; he has to either break the hash function $H$, or find a nonzero $z \in S_{\eta}^1$ for which $\textbf{A}.(z-\textbf{s}\textbf{e}_2)=0 \, \mathrm{mod}\, q$, i.e., find a solution of the Module-SIS problem given in Definition \ref{definition, Module-SIS} (with $m=1$ and $d=k^2$). This challenge is the aforementioned combined hard problem denoted as $\text{Adv}_{k,\eta}^{\text{H-MSIS}}$ (Note that as in Crystals-Dilithium, the infinity norm, rather than the Euclidean norm, has been considered to avoid the trivial solution $(q,0,\dots,0)^T$ for the underlying Module-SIS problem). More precisely, following Crystals-Dilithium \cite[Section 6.1]{CRYSTALS-Dilithium}, we simulate the signature queried by the adversary, to give the \textit{classical} reduction from Module-SIS to $\text{Adv}_{k,\eta}^{\text{H-MSIS}}$:
	
	Suppose the attacker $\mathcal{A}$ only has classical access to the function $H$ modeled as a random oracle. Assume that $\mathcal{A}$ gets the public key $(\textbf{A}, \textbf{P})$ and makes queries 
	\begin{equation*}
		H(\mu_1 || \textbf{A} \textbf{s} \textbf{e}_{2,1}), \dots, 	H(\mu_k || \textbf{A} \textbf{s} \textbf{e}_{2,t}),
	\end{equation*}
	for the messages $\mu_1,\dots,\mu_t$, where for every $i=1,2\dots,t$, $\textbf{e}_{2,i} \in S_{\eta}^k$ is the same as the error vector $\textbf{e}_2$ in the signing algorithm \ref{alg:signing}.
	% $\alpha_i \in S_{\eta}^1$, for $i=1,\dots,t$, is chosen by $\mathcal{A}$ (In fact, we assume that $\alpha_i=\textbf{s} \textbf{e}_i$ for some $\textbf{e}_i \in S_{\eta}^k$).
	 Then he receives randomly-chosen responses $h_1,\dots,h_t$ and outputs $\mu_i,\alpha_i,h_i$ for some $i=1,\dots,t$, such that $\alpha_i \in S_{\eta}^1$ and
	\begin{equation*}
		H(\mu_i || \textbf{A} \alpha_i)=h_i.
	\end{equation*}
	The reduction \textit{rewinds} $\mathcal{A}$ to the point where he made query $	H(\mu_i || \textbf{A} \textbf{s} \textbf{e}_{2,i})$ and reprograms the response to another randomly-chosen $h_i^{\prime}$. Then the \textit{forking lemma} states that the successful $\mathcal{A}$ has a chance of approximately $1/t$ to create a forgery on the same $i$.  In this case, he will output $\mu_i, \alpha_i^{\prime}, h_i^{\prime}$ satisfying
	\begin{equation*}
	H(\mu_i || \textbf{A} \alpha_i^{\prime})=h_i^{\prime}.
	\end{equation*}
	Since $h_i=h_i^{\prime}$, we will have 
	\begin{equation} \label{equation, A alphai}
		\textbf{A}.\left(\alpha_i -\alpha_i^{\prime}\right) =0,
	\end{equation}
	while $\alpha_i \neq \alpha_i^{\prime}$. On the other hand, all the coefficients of $\alpha_i - \alpha_i^{\prime}$ are small, and hence \eqref{equation, A alphai} gives a solution to Module-SIS.
	
	 %More formally, following Crystals-dilithium \cite[Section 6]{CRYSTALS-Dilithium}, the UF-CMA security of the proposed signature scheme in the \textit{Random Oracle Model (ROM)} can be stated as follows.
	%Now we have proved the following theorem (stated in the same format as in Crystals-schemes \cite{CRYSTALS-KYBER,CRYSTALS-Dilithium}).
	%Following Crystals-Dilithium \cite{CRYSTALS-Dilithium}, we also assume that the adversary has \textit{quantum access} to the hash function $H$ (modeled as a random oracle) and give a formal statement of the \text{UF-CMA} security as follows.
\end{proof}

\begin{remark}
Although the above reduction is \textit{non-tight}, yet it can be used to set the parameters, see \cite[Section 6]{CRYSTALS-Dilithium}. In addition, we have to consider the security in the \textit{quantum random oracle model (QROM)}. As stated in Crystals-Dilithium, although the above reduction does not transfer over the quantum setting, but the assumption $\text{Adv}_{k,\eta}^{\text{H-MSIS}}$ is \textit{tightly} equivalent, even under quantum reductions, to the SUF-CMA security of the proposed signature. Hence, in the above theorem, one can assume that the adversary $\mathcal{A}$ has a quantum access to the hash function $H$, see  \cite[Section 6]{CRYSTALS-Dilithium} for more details.
\end{remark}

%More precisely, our estimation of the security strength is based on the cost estimations of attacks against the underlying Module-LWE and Module-SIS problems as presented in crypt-analysis of Crystals-dilithum . More formally, following Crystals-dilithum \cite[Section 6]{CRYSTALS-Dilithium}, the UF-CMA security of the proposed signature scheme in the \textit{Random Oracle Model (ROM)} can be stated as follows. based on the hardness of Module-LWE and the combined hardness of  Module-SIS with the hash function $H$.  Moreover, we show that in the classical random oracle model, the proposed signature scheme is \textit{SUF-CMA} secure based on the hardness of the standard Module-LWE and Module-SIS lattice problems. In order to prove the security in the \textit{Quantum Random Oracle Model (QROM)} we give  non-tight reductions. Whereas there haven’t been any attacks taking advantage of the non-tightness
%of the reduction, since 

\section{Choosing Parameters and Security Estimation} \label{section, parameters}
%In this section, we determine some parameter sets satisfying the expected security levels and failure rates simultaneously. Recall that we have used the encoding-decoding functions introduced in NewHope-Simple algorithm \cite{alkim2016newhope} and obtained the error rate around $2^{-60}$. The NewHope algorithm \cite{NewHope-USENIX,alkim2016newhope} is a Ring-LWE based scheme whose parameters are $n \in \{512,1024\}$ and $q=12289$.
%Our digital signature scheme is based on Module-LWE which causes we choose smaller parameters, in comparison to NewHope ones, as $n=256$ and $q=3329$ for all claimed security levels. However, following NewHope \cite{NewHope-USENIX,alkim2016newhope}, we use the same parameter $\kappa=16$ for the centered binomial distribution $\Psi_{\kappa}$ to sample the secret and error terms. 

We use the ring $R_q=\frac{\mathbb{Z}_q[x]}{(x^n+1)}$ with $n$  a power-of-two, a favorable choice by many Ring/Module-LWE based schemes. We use the polynomial ring of degree $n=256$, for every parameter set, as the basic block of the Module-LWE problem. Similar to NewHope Scheme\cite{NewHope-USENIX,alkim2016newhope} we choose the modulus $q=12289$ and sample the Module-LWE secret and error terms by the centered binomial distribution $\psi_{\eta}$ of parameter $\eta=16$. Note that $q \equiv 1 \, (\mathrm{mod}\, n)$ which  enables the use of NTT over the ring $R_q$.
Also it must be noted that due to replacing the Ring-LWE with Module-LWE, we have decreased the ring dimension $1024$ in NewHope \cite{NewHope-USENIX,alkim2016newhope} to $256$,  yet we use the same modulus and the same sampling distribution as in NewHope to achieve the decoding failure probability $2^{-60}$. In other words, since we have used the encoding-decoding functions introduced in the NewHope-Simple algorithm, using the same modulus $q=12289$, the same binomial distribution $\psi_{16}$ and the ring dimensions $256k^2$ (with $k=2$), ensures us that the upper bound $2^{-60}$, as in NewHope \cite{NewHope-USENIX,alkim2016newhope}, works also for the failure probability in the proposed scheme, see Section \ref{section, Decoding failure probability}. Moreover, using these parameters, we can use the same security analysis as in NewHope-USENIX \cite[Section 6]{NewHope-USENIX}, and in comparison to Sharafi-Daghigh scheme \cite{sharafi2022ring}, we obtain much smaller failure probability and higher security levels at the cost of increasing the modulus $q$, see Table \ref{Table, security estimating}.

\begin{remark}
	Note that increasing the parameter $k$ will increase the security level but makes the scheme too inefficient. So we have determined only one set of parameters, whereas finding different encoding-decoding functions to achieve an appropriate failure rate with a smaller modulus $q$, might be an interesting challenge for future works. 
\end{remark}

\subsection{Security Estimation} 

As mentioned before, the security of the proposed digital signature scheme is based on the hardness of Module-LWE and Module-SIS problems. Any $\text{MLWE}_{k,D}$ instance for some distribution $D$ can be viewed as an LWE instance of dimension $(nk)^2$. Indeed the underlying $\text{MLWE}_{k,D}$ problem for the proposed digital signature can be rewritten as finding $\text{vec}(\textbf{s}), \text{vec}(\textbf{e}) \in \mathbb{Z}^{nk}$ from $(\text{rot}(\textbf{A}), \text{vec}(\textbf{P}))$ where $(\textbf{A},\textbf{P})$ is the public key as in Algorithm \ref{alg:key gen}, $\text{vec}(.)$ maps a vector of ring elements to the vector obtained by concatenating the coefficients of its coordinates, and $\text{rot}(\textbf{A}) \in \mathbb{Z}_q^{nk \times nk}$ is obtained by replacing all entries $a_{ij} \in R_q$ of $\textbf{A}$ by the $n \times n$ matrix whose $z$-the column is $\text{vec}(x^{z-1}.a_{ij})$.
Similarly, the attack against the $\text{MSIS}_{k,\eta}$ instance can be mapped to a $\text{SIS}_{nk,\eta}$ instance by considering the matrix $\text{rot}(\textbf{A}) \in \mathbb{Z}^{nk \times nk}$%The attacker may consider a subset of $w$ columns, and let the solution coefficients corresponding to the dismissed columns be zero, 
, see \cite[Appendix C]{CRYSTALS-Dilithium} for more details.

In order to do security estimation, we will not consider BKW-type attacks \cite{KF15} and linearization attacks \cite{AG11}, since based on our parameter set, there are only ``$(k+1)n$'' LWE samples available, see \cite[Section 5.1.1]{CRYSTALS-KYBER}. Hence the main challenge is security analysis against the \textit{lattice attacks}. All the best-known lattice attacks are essentially finding a nonzero short vector in the Euclidean lattices, using the Block-Korkine-Zolotarev (BKZ) lattice reduction algorithm \cite{SE94, CN11}. The algorithm BKZ proceeds by reducing a lattice basis using an SVP oracle in a smaller dimension $b$. The strength of BKZ increases with $b$, however, the cost of solving SVP is exponential in $b$ \cite[Appendix C]{CRYSTALS-Dilithium}. It is known that the number of calls to that oracle remains polynomial. Following NewHope \cite[Section 6.1]{NewHope-USENIX}, we ignore this polynomial factor, and rather evaluate only the \textit{core SVP hardness}, that is the cost of one call to an SVP oracle in dimension $b$, which is clearly a pessimistic estimation, from the defender's point of view.

There are two well-known BKZ-based attacks, usually referred to as
\textit{primal attack} and \textit{dual attack}. As explained in Crystals-Dilithium \cite[Appendix C]{CRYSTALS-Dilithium}, the primal attack for the proposed scheme is to find a short non-zero vector in the lattice
\begin{equation*}
	\Lambda=\{\textbf{x} \in \mathbb{Z}^d \, : \, \textbf{M}\textbf{x}=\textbf{0} \, \,  \mathrm{mod}\, q\},
\end{equation*}
where $\textbf{M}=(\text{rot}(\textbf{A})_{[1:m]} | \textbf{I}_m | \text{vec}(\textbf{P})_{[1:m]})$ is an $m \times d$ matrix with $d=nk+m+1$ and $m \leq nk$. By increasing the block size $b$ in BKZ, for all possible values $m$, the primal attack solves the SVP problem in the lattice $\Lambda$.

The dual attack consists in finding a short non-zero vector in the lattice
\begin{equation*}
	\Lambda^{\prime}=\{(\textbf{x},\textbf{y}) \in \mathbb{Z}^m \times \mathbb{Z}^d \, : \, \textbf{M}^T \textbf{x}+\textbf{y}=0 \, \mathrm{mod}\, q\},
\end{equation*}
where $\textbf{M}=(\text{rot}(\textbf{A}))_{[1:m]}$ is an $m \times d$ matrix  with $m \leq d=nk$. Similar to the primal attack, the dual attack solves the SVP problem in the lattice $\Lambda^{\prime}$, by increasing the block size $b$ in BKZ, for all possible values $m$. 

According to the parameters $n,k,q,\eta$, we have used the core SVP hardness methodology and the cost estimation of NewHope-USENIX \cite[Section 6]{NewHope-USENIX} against the primal and dual attacks, as presented in Table \ref{Table, security estimating}.

\begin{table}[!h] 
	\begin{center}
		\begin{tabular}{|cccc|ccc|c|c|p{1.8cm}|c|c|p{1.8cm}|}	
			\hline
			\multicolumn{4}{|c|}{\textbf{Parameters}} & \multicolumn{3}{|p{1.9cm}|}{\textbf{Key Sizes (in Bytes)}} & \multicolumn{3}{|c|}{\textbf{Primal attack}} & \multicolumn{3}{|c|}{\textbf{Dual attack}} \\
			\hline 
			$n$& $q$ &$k$ & $\eta$ & pk & sk & sig  & $m$ &	$b$  & Classical (Quantum) core SVP & $m$ & $b$ &  Classical (Quantum) core SVP \\
			\hline
			$256$ & $12289$ & $2$ & $16$ & 901.5 & 869.5 & 1771 & 1100 & 967 & 282 (256) & 1099 & 962 & 281 (255) \\ \hline
		%	$256$ & $12289$ & $3$ & $16$ & 1336 & 1304 & 2205 & 467 & 150 & 125 & 518 & 162 & 132 \\ \hline
		%	$256$ & $12289$ & $4$ & $16$ & 1771 & 1739  & 2640 & 780 & 230 & 192 & 834 & 235 & 210 \\ \hline
		\end{tabular} 
		\caption{Parameter sets for the proposed signature scheme. The parameters $n,q,k,\eta$ refer to the ring dimension, the modulus, the module-LWE dimension, and the parameter of the centered binomial distribution $\psi_{\eta}$, respectively. The notations ``pk'', ``sk'' and ``sig'' denote respectively the public key $(\mathbf{A},\mathbf{P})$, the secret key $\mathbf{s}$ and the signature $\sigma$ whose sizes are formulated in Remark \ref{remark, formula for size}. The value ``$b$'' denotes the block size and ``$m$'' denotes the number of used samples in the algorithm BKZ. As discussed in Section 5, the parameters $b$ and $m$, and consequently the estimated costs of Primal and Dual attacks, are taken from NewHope-USENIX \cite[Section 6]{NewHope-USENIX}.}
		\label{Table, security estimating}
	\end{center}
\end{table}

\section{Comparison to some related works}

Two of the most celebrated lattice-based signature algorithms are Crystals-Dilithium \cite{CRYSTALS-Dilithium} and FALCON \cite{Falcon-NIST}, selected for
NIST \footnote{National Institute of Standards and Technology} Post-Quantum
Cryptography Standardization  \cite{alagic2022status}. Crystals-Dilithium, one of our main references, is based on the Fiat-Shamir with aborts approach \cite{Lyu09} whose security is based on Module-LWE and Module-SIS problems. As stated in \cite[Section 4.4.2]{alagic2022status}, FALCON \footnote{Fast Fourier Lattice-based Compact Signatures over NTRU} is a lattice-based signature scheme utilizing the ``hash-and-sign'' paradigm. FALCON follows the GPV framework, introduced by Gentry, Peikert, and Vaikuntanathan \cite{GPV08} and builds on a sequence of works whose aim is to instantiate the GPV approach efficiently in NTRU lattices \cite{stehle2011making,ducas2014efficient,ducas2016fast} with a particular focus on the compactness of key sizes. There is also another lattice-based digital signature, namely ``qTESLA'' \cite{bindel2018submission}, among the second-round (but not finalist) candidate algorithms in the NIST post-quantum project. qTESLA is a Ring-LWE-based digital signature in which signing is done using hash functions and Fiat-Shamir with aborts technique.

In addition to the above algorithms in the NIST post-quantum project, for comparison, we have tried to consider those digital signatures, including Sharafi-Daghigh scheme \cite{sharafi2022ring}, BLISS \cite{DDLL13}, and GLP \cite{guneysu2012practical}, that are most relevant to the proposed scheme. The comparison results are presented in Table \ref{Table, comparison}.

\begin{table}
	\begin{center}
		\begin{tabular}{|p{3.8cm}|c|c|c|c|c|}
\hline
	\textbf{Algorithm} &	\multicolumn{3}{|c|}{\textbf{Key Sizes (in Bytes)}} & 	\multicolumn{2}{|c|}{\textbf{Security level (bits)}} \\
\cline{2-6}
		 & pk & sk & sig & Classic & Quantum \\ \hline
		 Dilithium \cite{CRYSTALS-Dilithium} (NIST Security Level 3) & 1952 & 4016 & 3293 & 182 & 165 \\ \hline
		  Dilithium \cite{CRYSTALS-Dilithium} (NIST Security Level 5) & 2592 & 4880 & 4595 & 252 & 229 \\ \hline
		  FALCON \cite{Falcon-NIST} (NIST Security Level 1) & 897 & 1281 & 666 & 133 & 121 \\ \hline
		  FALCON \cite{Falcon-NIST} (NIST Security Level 5) & 1793 & 2305 & 1280 & 273 & 248 \\ \hline
		  qTESLA \cite{bindel2018submission} (NIST Security Level 1) & 14880 & 4576 & 2848 & 132 & 123 \\ \hline 
		   qTESLA \cite{bindel2018submission} (NIST Security Level 3) & 39712 & 12320 & 6176 & 247 & 270 \\ \hline
		   TESLA-2 \cite{ABB+17} & 21799  & 7700  & 4000 & 128 & 128 \\ \hline
		    TESLA\#-II \cite{barreto2016sharper} & 7168 & 4224 & 3488 & 256 & 128 \\ \hline
		 Ring-TESLA-II \cite{chopra2016improved} & 1500 & 3300 & 1800 & 80 & 73 \\ \hline 
		   BLISS-IV \cite{DDLL13} & 875 & 375 & 812.5 & 192 & -- \\ \hline
		   GLP (Set II) \cite{guneysu2012practical} & 3125 & 406 & 2350 & $> 256$ & -- \\ \hline
		   Sharafi-Daghigh \cite{sharafi2022ring} & 1032 & 520 & 1565 & -- & 92 \\ \hline
		   Proposed scheme & 901.5 & 869.5 & 1771 & 281 & 255 \\ \hline
		\end{tabular}
		\caption{Comparison of the parameters of the proposed signature scheme to some well-known lattice-based signature algorithms. The notations ``pk'', ``sk'', and ``sig'' denote public key, secret key, and signature, respectively. Also, blanks mean that we could not find the related values in the corresponding references.}
		\label{Table, comparison}
	\end{center}
\end{table}

\section*{conflict of interest}

On behalf of all authors, the corresponding author states that there is no conflict of interest.

%\pagebreak
%%\newpage

\clearpage 

\bibliographystyle{alpha}

\bibliography{library-MLWE}

%%\bibliography{library}

\end{document}